\documentclass[12pt,draftcls,onecolumn]{IEEEtran}
\usepackage{amsfonts,epsfig,amsmath,latexsym,amssymb,amscd,multirow,graphicx,geometry,lscape,amsmath,amssymb,bm,pifont,graphicx,amssymb,amsmath}
\usepackage{amsthm}
\usepackage[]{cite}
\usepackage{cancel,algorithm,verbatim,array,multicol,palatino,amsfonts}
\usepackage[usenames,dvipsnames]{xcolor}
\usepackage{graphicx}
\usepackage{lineno}
\usepackage{subfig}
\usepackage[normalem]{ulem}
\usepackage{dsfont}
\usepackage{framed}
\usepackage{algpseudocode}

\newtheorem{lemma}{\noindent \textbf{Lemma}}
\newtheorem{stylized}{\noindent \textbf{Stylized Fact}}
\newtheorem{proposition}{\noindent \textbf{Proposition}}

\newtheoremstyle{remarkstyle}
  {\topsep}   
  {\topsep}   
  {\slshape}  
  {0pt}       
  {\bfseries} 
  {.}         
  {5pt plus 1pt minus 1pt} 
  {}          

\theoremstyle{remarkstyle}
\newtheorem{remark}{\noindent \textbf{Remark}}
\newtheorem{definition}{\noindent \textbf{Definition}}

\modulolinenumbers[1]

\begin{document}
\title{Market Mechanism Design for Profitable On-Demand Transport Services}
\author{Malcolm Egan$^{1,*}$ and Michal~Jakob$^{1}$

\begin{center}
{\footnotesize {\
\textit{
$^{1}$ Agent Technology Center, Faculty of Electrical Engineering, Czech Technical University in Prague, Czech Republic.\\[0pt]
Email: \{malcolm.egan,michal.jakob\}@agents.fel.cvut.cz\\[0pt]
$^{*}$ Corresponding author.
} } }
\end{center}
}

\maketitle


\begin{abstract}
%

On-demand transport services in the form of dial-a-ride and taxis are crucial parts of the transport infrastructure in all major cities. However, not all on-demand transport services are equal. In particular, not-for-profit dial-a-ride services with coordinated drivers significantly differ from profit-motivated taxi services with uncoordinated drivers. As such, there are two key threads of research for efficient scheduling, routing, and pricing for passengers: dial-a-ride services (first thread); and taxi services (second thread). Unfortunately, there has been only limited development of algorithms for joint optimization of scheduling, routing, and pricing; largely due to the widespread assumption of fixed pricing. In this paper, we introduce another thread: profit-motivated on-demand transport services with coordinated drivers. To maximize provider profits and the efficiency of the service, we propose a new market mechanism for this new thread of on-demand transport services, where passengers negotiate with the service provider. In contrast to previous work, our mechanism jointly optimizes scheduling, routing, and pricing. Ultimately, we demonstrate that our approach can lead to higher profits, compared with standard fixed price approaches, while maintaining comparable efficiency.
\end{abstract}
\newpage

\newpage

\section{Introduction}\label{Introduction}

Call a taxi in any major city and it will often arrive within minutes. Despite the success and widespread use of taxis and other on-demand transport services, there is room for improvement: higher profits; reduced prices; and even lower waiting times, targeted at passengers that have the highest demand for the service. Fortunately, the ubiquity of mobile internet and secure online financial transactions has opened the way for highly efficient on-demand transport; able to ensure that the right driver transports the right passenger at the right time.

At present, efficiently operating a fleet of on-demand vehicles remains difficult. The cause is the coupling between three key sub-problems: which passengers should be serviced by each vehicle \textit{(routing)}; what time should each passenger be picked up \textit{(scheduling)}; and how much each passenger should be charged \textit{(pricing)}. In effective on-demand transport services, the sub-problems should be solved jointly, not be decoupled.

Unfortunately, there are few computational techniques to solve the three on-demand sub-problems jointly; despite over three decades of research on related vehicle routing problems. While this might seem surprising, there is good reason: there are in fact two distinct research threads---each addressing a different niche.

The first thread is dial-a-ride services---targeted at the niche of elderly and disabled transport, as well as in low demand regions. Vehicles in fleets offering this type of service are coordinated by a single provider. Each vehicle collects passengers within requested pick-up time intervals and drops each passenger off before a requested drop-off time. Dial-a-ride services are typically heavily subsidized by governments, due to the important role they play for vulnerable members of the community; for instance, by taking an elderly woman to the hospital for a check-up. These subsidies have an important repercussion: dial-a-ride services are not-for-profit. As such, the aim of providers is to minimize costs; as opposed to maximizing profits. This means that the pricing sub-problem is not considered, and the most popular formalization known as the dial-a-ride problem focuses on the routing and scheduling sub-problems. Both centralized (see \cite{Cordeau} for an extensive survey) and decentralized agent-based \cite{Barbucha,Cubillos} approaches have been proposed.


The second thread is taxi and private hire services. Taxi services contrast with dial-a-ride services in two key ways: they are profit-motivated; and vehicles are not heavily coordinated by a single service provider. This is because the drivers are self-interested and unable to easily determine the current locations and destinations of other nearby taxis. While taxi services are profit-motivated, the most common pricing strategy is to use a fixed price-rate; i.e., the price scales as a---usually linear---function of the distance (see \cite{Yang2010} for pricing with a non-linear function of distance). As such, the price-rate does not factor in the number and demand\footnote{The notion of demand here can be easily misconstrued. We mean demand relative to factors such as cost of living and time of day. We do not mean in terms of need; in particular, we believe that pricing should not be manipulated in emergency situations.} of passengers that have ordered a ride. Moreover, the fixed price-rate means that the pricing sub-problem is decoupled from routing and scheduling; instead, the focus is on reducing waiting times and travel distance (and hence reducing costs) \cite{Seow,Bai,Yang2011,Balan,Glaschenko}. We note that the approach in \cite{Oren2014} does optimize pricing for single taxi operation via dynamic programming; however, scheduling and routing for a fleet of taxis is not considered.

In essence, algorithms for the joint solution of the three on-demand sub-problems have not been developed in either of these well-established threads of on-demand transport research. As such, in this paper, we introduce a new thread---made practical with recent technological developments---where joint solution is feasible.

\subsection{A New Thread}

The aim of the new thread of on-demand transport services is profitability and efficiency when drivers are coordinated (as in dial-a-ride services) and providers are profit-motivated (as in taxi services). It is worth noting that this category only recently arose as more drivers and passengers have adopted internet-enabled smartphones, allowing for providers to easily keep track of demand statistics and for passengers to directly pay the provider (instead of the driver).

Looking from the computational perspective, standard algorithms for dial-a-ride or taxi services cannot be directly applied. This is due to the fact that the first two threads decouple pricing from routing and scheduling. As such, to solve the three sub-problems jointly, new computational techniques are required.

In this paper, we propose a new market-based approach for the new thread of on-demand transport services. In fact, we jointly solve the three sub-problems: routing; scheduling; and pricing. There are two fundamental aspects of our approach: a new passenger model; and a new market mechanism. The passenger model goes beyond the standard approach in the dial-a-ride problem by providing a realistic probabilistic model for each passenger's expectations---ultimately allowing the provider to tailor journey offers to the passengers that request the service. Our new approach significantly extends our initial work in \cite{Egan2014} by enriching the passenger models, and improving the joint scheduling, routing and pricing algorithm to allow for deviations from requested journeys.

\subsection{Modeling Passenger Expectations}

In traditional approaches to the dial-a-ride problem, as well as taxi routing and scheduling, a passenger is simply a request; that is, pick-up and drop-off times, and locations. By viewing a passenger in this way it is not possible to optimize the price for each passenger. This is due to the fact that the service provider must also account for passenger preferences; in other words, how likely a passenger is to accept a journey.

To enrich the passenger model, we account for the probability that a passenger will accept an offer, on top of her request. Two key factors are considered: the price of the offer; and the journey deviation. The journey deviation corresponds to the difference between the time that passengers request pick-up and drop-off, and the actual times. This passenger model forms the basis for price optimization in our new mechanism for scheduling, routing, and pricing.

\subsection{A Market Mechanism for the New Thread}

At its heart, the new thread of on-demand transport services is multiple, independent, passengers ordering transportation from a service provider with multiple vehicles. A natural way forward is to use a market, as markets, by definition, exchange goods or services for money.

In this paper, we propose a  market mechanism for the new thread of on-demand transport services---profit-motivated service providers with coordinated drivers. Our market mechanism is designed to jointly route and schedule vehicles, and price passengers. In particular, we introduce a four-stage mechanism, initiated by the service provider generating an offer for passengers, and ending with passengers making a final decision of whether to accept or reject.

It is worth pointing out the key difference between our market mechanism and other negotiation mechanisms used for on-demand transport services. That is, our mechanism ultimately schedules, routes, and \textit{prices} passengers. This is not the case in other approaches. The reason is that the other negotiation mechanisms are used for, in a decentralized fashion, scheduling and routing vehicles. However, there is no consideration of how to optimally price passengers. On the other hand, we are able to optimize passenger pricing (in addition to routing and scheduling vehicles) due to our enriched passenger model, which captures passenger expectations.

\subsection{Key Contributions}

In this paper, we introduce a new thread of on-demand transport services---profit-motivated providers and coordinated drivers---and a market mechanism to route, schedule, and price passengers. We summarize our key contributions as follows:
\begin{enumerate}
\item \textbf{Agent-based passenger modeling:} We introduce new models for passengers that are enriched to include expectations for price and deviations from requests, on top of the standard request-based model.
\item \textbf{Market mechanism:} We propose a new market-mechanism to jointly schedule, route, and price passengers. Offers are generated for passengers via an expected profit maximization algorithm. We also analyze the effect of varying the time between when the mechanism is run, which leads to closed-form analytical expressions that provide key design insights.
\item \textbf{Simulation results:} The on-demand transport services considered in this paper are a new category and as such there is a genuine need for a business case. To this end, we perform a simulation study to evaluate the potential for profit, and also the efficient service of passengers. We demonstrate that incorporating passenger expectations in our new passenger models can in fact improve profitability over standard fixed price-rate approaches, while maintaining comparable efficiency.
\end{enumerate}

\section{Modeling Agents}\label{sec:agent_model}

Consider the network consisting of a single on-demand transport service provider and $N$ passengers. The service provider owns a fleet of $K$ \textit{unit capacity} vehicles that all start and finish their journeys at a common depot; each vehicle traveling with average velocity $\nu$. Each passenger has requested pick-up and drop-off locations, which are represented by elements from the set of vertices $V$ in a directed graph $G$. The directed graph $G$ represents the underlying road network. As such, the set of edges $E$ in $G$ represent direct routes between locations in $V$.

Associated to each edge $e \in E$ (corresponding to direct routes between locations) are:
\begin{enumerate}
\item a start location $u \in V$;
\item an end location $w \in V$;
\item a cost $c_e \in [0,\infty)$ to the service provider to traverse edge $e \in E$;
\item and an edge traversal time $\tau_e \in \mathbb{Z}_+$.
\end{enumerate}
The edge cost $c_e$ and the edge traversal time $\tau_e$ are found during pre-processing where the service provider solves the shortest path problem between $u$ and $v$ on the underlying road network. Note that when edge $e$ connects the vertices $u$ and $w$, the traversal time is denoted by $\tau_{u,w}$. This model is appropriate when drivers are salaried or pay a fixed commission, which can be easily incorporated into the edge cost $c_e$.

So far, we have simply described the basic service provider and passenger model used to model dial-a-ride services. In order to enrich the model for the new thread on-demand services, we need to introduce two new features to the model:
\begin{enumerate}
\item passengers capable of making a decision whether or not to accept a journey offer;
\item and a service provider capable of evaluating the probability a passenger will accept her offer.
\end{enumerate}

\subsection{Passengers}

The first step in modeling passengers is descriptive. In particular, we identify two stylized facts, which we believe hold for passengers using on-demand transport services. The concept of stylized facts has been widely used in computational macroeconomics as a means of validating descriptive models of real markets \cite{Cont2001}. Based on the stylized facts, we develop a passenger policy, which is used to model how passenger decide whether or not to accept a journey.

Before introducing the stylized facts and the passenger policy, we define the parameters that determine passenger behavior. First, a request from passenger $i$ consists of:
\begin{enumerate}
\item a pick-up location $v_{i,p} \in V$;
\item a drop-off location $v_{i,d} \in V$;
\item a pick-up time interval $(a_i,b_i) \in \{0,1,2\ldots\} \times \{0,1,2,\ldots\}$, with $a_i \leq b_i$;
\item and a latest drop-off time $l_i \in \{0,1,2,\ldots\}$, with $l_i > b_i$.
\end{enumerate}
We note that we only explicitly consider the latest drop-off time (instead of a drop-off time interval) as the earliest pick-up time determines the earliest drop-off time.

In response to the requests, the service provider offers a journey to each passenger, which consists of two components: the deviation of the journey from the request; and the price of the journey. More precisely, the deviations are defined as follows.
\begin{definition}\label{def:deviation}
Let $T_i$ be the actual pick-up time and $L_i$ be the actual drop-off time, for passenger $i$. The \textit{pick-up interval deviation}, denoted by $\gamma_{p,i}$, is defined as
\begin{align}
\gamma_{p,i} = \left\{
             \begin{array}{ll}
               a_i - T_i, & \textrm{if}~T_i < a_i \\
               T_i - b_i, & \textrm{if}~T_i > b_i \\
               0, & \textrm{otherwise.}
             \end{array}
           \right.
\end{align}
Similarly, the \textit{drop-off time deviation}, denoted by $\gamma_{d,i}$, is defined as
\begin{align}
\gamma_{d,i} = \left\{
             \begin{array}{ll}
               L_i - l_i, & \textrm{if}~L_i> l_i \\
               0, & \textrm{otherwise.}
             \end{array}
           \right.
\end{align}
The \textit{deviation} is then given by $\delta_i = \gamma_{p,i} + \gamma_{d,i}$.
\end{definition}

The next step towards the passenger policy is to introduce two stylized facts for passenger behavior. Stylized facts---widely used in computational economics \cite{Tesfatsion2001}---are important as they provide a means of justifying the passenger policy as a plausible description of real passenger decision-making relevant to on-demand transport.

Our first stylized fact is as follows.
\begin{stylized}\label{sf:1}
The maximum price a given passenger will pay for a journey and her maximum deviation $\delta$ do not vary significantly for passengers that regularly use on-demand transport services. In other words, these parameters will in general vary from passenger to passenger, but not for the same passenger.
\end{stylized}

We can justify this stylized fact by observing that passengers change transportation habits with great difficulty \cite{Verplanken2006}. As such, we expect that regular on-demand users will have a well-defined maximum price they are prepared to pay.

Next, we have the second stylized fact.
\begin{stylized}\label{sf:2}
The probability that a passenger will accept an offer decreases when either the price increases with the deviation fixed, or the deviation increases with the price fixed.
\end{stylized}
More colloquially, this stylized fact is a formal statement of the intuitive notion that if there is a better deal, more passengers will accept.

Based on the first stylized fact, we propose the following descriptive model for passenger policies used to determine whether or not to accept an offer:
\begin{enumerate}
\item If $r < r_{\max}$ and $\delta < \delta_{\max}$, then the passenger will accept;
\item Otherwise, then the passenger will reject.
\end{enumerate}
Observe that this policy ensures that the first stylized fact holds. Moreover, our descriptive model does not violate the second stylized fact, as the policy is concerned with individual passenger decisions and the stylized fact is concerned with the aggregate. We also note that in principle $\delta_{\max}$ can depend on $r_{\max}$.

\subsection{Service Provider}

The aim of the service provider is to schedule, route and price passengers to maximize its expected profit. In the literature, the price-rate has been commonly assumed to be fixed, and only the scheduling and routing were optimized. This means that there is no uncertainty and the expected profit can be maximized by minimizing the cost. The key new feature that we introduce for the service provider is a probabilistic model of each passenger. As we show in Section~\ref{sec:proposed_approach}, this extra information improves the expected profit over standard approaches.

We now detail the passenger model that the service provider uses. Importantly, we also show that it satisfies both the stylized facts, which means that it is a plausible probabilistic descriptive model of passengers.

In order for the service provider to infer demand at a given price, it requires the probability each passenger will accept her offer. It is necessary to consider the \textit{probability} an offer is accepted as the service provider does not perfectly know the maximum deviations and price that any given passenger will accept. In particular, the probability passenger $i$ accepts her offer is given by
\begin{align}\label{eq:prob_accept}
\mathrm{Pr}(i~\mathrm{accept}) = \mathrm{Pr}(\delta_i \leq \delta_{i,\max},r_i \leq r_{i,\max}).
\end{align}

\begin{remark}
We emphasize that the realizations of the maximum deviation $\delta_{i,\max}$ and the maximum price rate $r_{i,\max}$ are not known to the service provider, only to passenger $i$.
\end{remark}

To obtain the probability that any given passenger accepts, we assume that the service provider knows the joint probability density function $f(\delta_{i,\max},r_{i,\max})$. Our assumption that the service provider has statistical knowledge of $(\delta_{i,\max},r_{i,\max})$ ensures that each passenger is not always charged at the maximum possible price she is prepared to pay--realistic in competitive profit-motivated on-demand transport services. On the other hand, statistical knowledge is enough to enable the service provider to optimize the expected profit, as we detail in Section~\ref{sec:proposed_approach}.

The density function $f(r_{i,\max},\delta_{i,\max})$ will typically depend on factors such as time of day, or the location of the service region. We focus on scenarios where the maximum deviations and price-rate are independent, which occur when the factors determining the maximum deviations and the financial factors affecting the price-rate are unrelated. An example is the taxi spot market, where only a small deviation is acceptable. In these scenarios, the density function is separable; i.e.,
\begin{align}
f(r_{i,\max},\delta_{i,\max}) = f_r(r_{i,\max})f_{\delta}(\delta_{i,\max})
\end{align}
We model $f_r(r_{i,\max})$, $f_{\delta}(\delta_{i,\max})$ via the scaled Beta distribution with parameters $(\alpha_r,\beta_r)$, $(\alpha_{\delta},\beta_{\delta})$ respectively. The reason for this is that the Beta distribution is a flexible distribution, which generalizes a wide variety of distributions with bounded support. The density functions are given by
\begin{align}\label{eq:beta_density}
f_r(r_{i,\max}) &= \frac{1}{r_uB(\alpha_r,\beta_r)} \left(\frac{r_{i,\max}}{r_u}\right)^{\alpha_r-1}\left(1 - \frac{r_{i,\max}}{r_u}\right)^{\beta_r - 1},\notag\\
f_{\delta}(\delta_{i,\max}) &= \frac{1}{\delta_uB(\alpha_{\delta},\beta_{\delta})} \left(\frac{\delta_{i,\max}}{\delta_u}\right)^{\alpha_{\delta}-1}\left(1 - \frac{\delta_{i,\max}}{\delta_u}\right)^{\beta_{\delta} - 1},
\end{align}
where $B(\alpha,\beta)$ is the Beta function, and the densities have support $[0,r_u]$ for $r_{i,\max}$ and $[0,\delta_u]$ for $\delta_{i,\max}$.

It is easy to see that both stylized facts are features of the passenger model that the service provider uses. In particular, as $r_{i,\max}$ and $\delta_{i,\max}$ increase, the probability that a passenger will accept is reduced. This means that both the descriptive passenger models we have detailed will exhibit the features described by the stylized facts.

\subsection{On Prediction}

So far, we have argued that our descriptive passenger models satisfy the stylized facts exhibited by real passengers. However, at this point these stylized facts have only been shown to hold under current conditions; that is, the present fixed price-rate approach. The remainder of this paper is concerned with a new routing, scheduling, and pricing approach. Clearly this is a structural change and as such, it is necessary to justify that the stylized facts still hold. This is due to the fact that we are predicting the performance of a socio-technical system under structural changes; known to be notoriously hard to do in many economic settings \cite{Tesfatsion2001}.

Fortunately, the stylized facts (Stylized Facts~\ref{sf:1} and \ref{sf:2}) are independent of how the service provider performs scheduling, routing, and pricing. This suggests that the stylized facts should hold even when the service provider changes the underlying algorithms and indeed the market. Despite this, it is possible that the parameters of the distributions for the maximum price-rate and deviation may vary when the pricing algorithm is changed. These variations can potentially be overcome by updating the parameters $\alpha_r,\beta_r,\alpha_{\delta},\beta_{\delta}$ to appropriately model the passengers' preferences.

The invariance of the stylized facts to service provider scheduling, routing, and pricing, suggests that key features arising from analysis of our model will be consistent with real-world practice.

\section{Our Proposed Market Mechanism}\label{sec:proposed_approach}

In this section, we propose a new market-based approach for scheduling, routing, and pricing in the new thread of on-demand transport services. We first detail the desiderata that our design should fulfill. We then overview the proposed market mechanism, and detail step-by-step the interactions between the service provider and each passenger.

\subsection{Design Desiderata}

The design of on-demand transport services is constrained by the physical (i.e., the vehicle fleet) and financial resources of the service provider, and the expectations of passengers. Ultimately, these constraints determine whether or not the provider is financially viable.

To ensure that the resource constraints of the provider and the expectations of passengers are satisfied, the design of our market mechanism is guided by three key desiderata:
\begin{enumerate}
  \item The service provider should be profitable.
  \item The passengers that desire the service the most should obtain it.
  \item There should be a simple interface between passengers and the service provider.
\end{enumerate}
Our first two desiderata are to ensure that the service provider is profitable and that vehicles are allocated to the passengers that most value the service, while the third desiderata brings our design in line with current trends that simplify access to transportation services.

We point out that our second desiderata is closely related to the standard notion of efficiency in mechanism design \cite{Shoham2009}; that is, an efficient mechanism allocates the service to the passengers that are prepared to pay the most for it. In particular, we can formalize the value of a journey to a given passenger as follows.
\begin{definition}\label{def:efficiency}
  Denote the the maximum price for passenger $i$ as $p_{i,\max}$. If the set of serviced passengers is $S$, then the efficiency of our market mechanism is then defined as
  \begin{align}
    \mathcal{E} = \sum_{i \in S} p_{i,\max}.
  \end{align}
\end{definition}

Our third desiderata is to bring our mechanism in line with current trends towards simplifying access to transpotation. The key impact of this desiderata on our approach is that passengers are not required to price their own journey; instead, the service provider always generates the first offer. This means that passengers do not need to be aware of the behavior of other passengers or how the service provider allocates vehicles. The only decision each passenger needs to make is whether the journey offer is acceptable or not, which in our model is determined by the deviation and price of the offered journey.

\subsection{Overview}

The three key design desiderata motivate a market mechanism, where passengers are not required to price their own journeys and the service provider generates offers that maximize its profits. In our approach, we propose the following mechanism structure:
\begin{enumerate}
  \item The service provider makes each passenger an initial offer. The offer to each passenger is based on vehicle allocations and passenger pricing done by the service provider to maximize its average profit.
  \item Each passenger responds to the initial offer by either rejecting or conditionally accepting the offer. Conditional acceptance is a contract between the passenger and the service provider, which requires the passenger to pay the amount offered unless either the price is raised or the deviation from the requested journey increases.
  \item The service provider computes the final vehicle journey plans and passenger pricing.
  \item The passenger either rejects the offer (due to the service provider breaking the contract) or unconditionally accepts, where the passenger is required to pay the provider for the service.
\end{enumerate}

A key aspect of the mechanism is that the service provider generates initial offers for each passenger via optimization of the average profit. There are two sets of variables associated to the optimization problem. First is the set of passengers allocated to the same vehicle, which leads to the allocation $C = \{C_1,\ldots,C_K\}$ with $C_i$ corresponding to the passengers allocated to the $i$-th vehicle. Second is the price-rate, $r$ that each passenger will pay. Formally, the optimization problem is
\begin{equation}\label{eq:expected_profit}
\begin{aligned}
& \underset{C_1,\ldots,C_K,r}{\text{maximize}}
& & \sum_{S \subset N} \left(\sum_{i \in S} rR_i - c_i\right) \prod_{i \in S} \mathrm{Pr}(i~\mathrm{accept}) \prod_{j \in S^c} \left(1 - \mathrm{Pr}(j~\mathrm{accept})\right) \\
& \text{subject to}
& & 0 \leq r \leq r_u,
\end{aligned}
\end{equation}
where $\mathrm{Pr}(i~\mathrm{accept})$ is given by (\ref{eq:prob_accept}), $r$ is the price-rate offered to each passenger, and $c_i$ is the cost of servicing passenger $i$.

It is important to note that this optimization problem generalizes the standard formulation for dial-a-ride services. In particular, we optimize over the pricing (encoded in $r$) and the scheduling and routing (encoded in the sets $C_1,\ldots,C_K$). In contrast, minimum cost scheduling and routing for dial-a-ride services does not consider pricing; i.e., the price-rate $r$ is assumed to be fixed. As such, there is potential for higher expected profit than simply using a minimum cost approach.

We also point out that irrespective of how the scheduling and routing is performed, using our approach in (\ref{eq:expected_profit}) will always lead to a higher expected profit. This is particularly important in the case that it is not tractable to optimally solve the clustering problem and heuristics are required.

While optimally scheduling, routing, and pricing passengers will yield a higher expected profit, it is not straightforward to solve. In particular, the problem in (\ref{eq:expected_profit}) is difficult for two reasons:
\begin{enumerate}
  \item The objective is generally nonlinear and also non-convex.
  \item The number of sets $C_1,\ldots,C_K$ to be searched is even greater than in the standard approaches as there are no feasibility constraints that constrain the passengers able to be serviced by a given vehicle. This occurs because we allow deviations from passenger requests.
\end{enumerate}

To alleviate the two difficulties in solving (\ref{eq:expected_profit}), we adopt the ``cluster-then-price" strategy; that is, we solve the problem (\ref{eq:expected_profit}) in two stages. In the first stage we cluster passengers into groups that are all served by the same vehicle. We note that the cluster formation is not a straightforward extension of standard approaches (e.g., \cite{Dumas}) due to the profit-based objective. We detail our clustering algorithm in Section~\ref{sec:stage_1}. We then show how the price-rate offered to each passenger is obtained by pricing the passengers based on the clustering. The remainder of section details the interactions between the service provider and passengers in our mechanism--including how final vehicle allocations and prices are computed.

Our proposed market mechanism is summarized in Fig.~\ref{fig:overview}, where dependencies are illustrated by arrows between each stage.

\begin{figure}[!h]
\centerline{\includegraphics[height=4in]{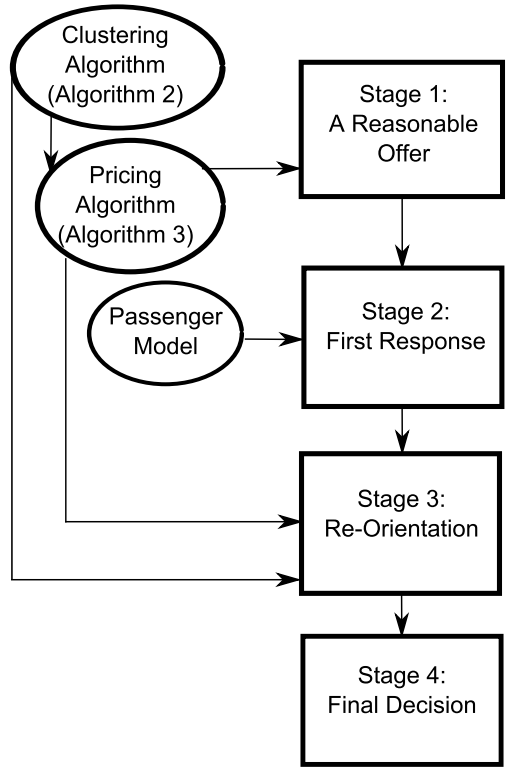}}
\caption{Overview of the proposed market mechanism for pricing and vehicle allocation in on-demand transport networks.} \label{fig:overview}
\end{figure}

\subsection{Stage 1: A Reasonable Offer}\label{sec:stage_1}

In the first stage of the mechanism, the service provider generates an offer for each passenger $i$. The offer consists of:
\begin{enumerate}
\item maximum deviation, $\gamma_i$;
\item and the price, $p_i = rR_i$, where $R_i$ is the distance between pick-up and drop-off locations for passenger $i$.
\end{enumerate}

To allow deviations from passenger requests during clustering, we introduce a probabilistic feasibility constraint, which is closely related to the probability that a passenger will accept the journey. The constraint is defined as follows.
\begin{definition}\label{def:eps_feasible}
A journey is $\epsilon$-feasible for passenger $i$, if
\begin{align}
\mathrm{Pr}(\delta_{\max} \geq \delta_i) \geq 1 - \epsilon,
\end{align}
where $\delta_{\max}$ is a random variable representing the unknown maximum deviation for passenger $i$. Moreover, a cluster is said to be $\epsilon$-feasible if each passenger in the cluster is $\epsilon$-feasible, with the set of $\epsilon$-feasible clusters denoted by $\mathcal{F}_{\epsilon}$.
\end{definition}

Intuitively, the notion of $\epsilon$-feasible allocations generalizes the standard notion of feasible clustering. That is, standard hard constraint clustering is $0$-feasible; i.e., each passenger must have a feasible journey with probability one. Importantly, $\epsilon$-feasible allocations are a key feature in our approach, as by relaxing the feasible constraints it is possible for the service provider to offer journeys with deviations from requests to passengers. Ultimately, this can allow the provider to service more passengers with a single vehicle, which (as we show in Section~\ref{sec:simulation}) leads to higher profits and efficiency of the market mechanism.

\textit{Passenger clustering: }We now develop a clustering algorithm that minimizes the cost of clusters subject to the probabilistic feasibility constraint detailed in Definition~\ref{def:eps_feasible}. In order to form passenger clusters, we determine whether or not the addition of a passenger to a cluster is $\epsilon$-feasible. This is achieved by computing the times that the inserted passenger will be picked up and dropped off. To reduce the complexity of the insertion algorithm, a passenger can be inserted only if the new passenger does not affect the deviations of the passengers already in the cluster, which means that previously clustered passenger pick-up and drop-off times are not changed after the passenger has been allocated to a cluster.

To illustrate, consider a passenger $0$ to be inserted into cluster $C_j = \{k_1,\ldots,k_{|C_j|}\}$, where passengers in $C_j$ are in the order that they are serviced. There are three situations to consider: insert passenger $0$ before passenger $k_1$; insert passenger $0$ between passengers $k_l$ and $k_{l+1}$, with $1 \leq l < |C_j|$; and insert passenger $0$ after passenger $k_{|C_j|}$.

To determine whether a passenger can potentially be inserted into the cluster, there are three checks that need to be performed:
\begin{enumerate}
  \item[\textbf{C1}:] Can the passenger be inserted without changing the deviation of previously clustered passengers?
  \item[\textbf{C2}:] If C1 holds, what are the pick-up and drop-off times with the minimum deviation (i.e., $\delta_i$)?
  \item[\textbf{C3}:] Is the cost of traveling to the passenger's pick-up location and from the drop-off location less than any other previously checked cluster?
\end{enumerate}
The checks are performed using the insertion algorithm detailed in Algorithm~\ref{alg:insertion}. In particular, the algorithm attempts to insert passengers such that the new passenger does not change the deviations of previously clustered passengers---corresponding to \textbf{C1}. In the case that it is possible to insert the passenger without affecting other clustered passengers, the insertion algorithm determines whether the passenger is $\epsilon$-feasible and computes the pick-up and drop-off times to minimize the deviation (implementing \textbf{C2}). Finally, the cost is computed and compared with previously checked clusters, and the lowest cost, $c_{best}$, insertion is updated (based on \textbf{C3}).

The minimum deviation pick-up and drop-off times (for \textbf{C2}) are computed using simple inequality tests due to the fact that the objective is linear; i.e., it is $\delta = \gamma_{p} + \gamma_{d}$ (from Section~\ref{sec:agent_model}). In particular, there are three types of insertion tests: before the first passenger in the cluster; between passengers $j$ and $j+1$; and after the last passenger. To illustrate, consider the potential insertion of passenger $0$ between passengers $j$ and $j+1$. First, the new passenger's journey must fit. This means that $L_j + \tau_{j,0} + \tau_{0,0} + \tau_{0,j+1} < T_{j+1}$. Next, we choose the minimum deviation insertion. This is done by checking whether $T_{j+1} - \tau_{0,j+1} \leq l_0$, which means that the deviation can only be caused by the pick-up time. The pick-up time is then chosen to minimize the deviation. The case where $T_{j+1} - \tau_{0,j+1} > l_0$ can be treated similarly.

The cost is computed as
\begin{align}\label{eq:cost_ins}
c_{i,ins} = c_{l-1,i} + c_{i,l},
\end{align}
where passenger $i$ is inserted between passenger $l-1$ (if $l = 0$, then this corresponds to the depot), and passenger $l$ (this may also be the depot).

\begin{algorithm}[!t]
\begin{algorithmic}
\Procedure{Insertion}{$c_{best}$,$i$,$C_j$}
\State Set passenger to be inserted as $i$.
\State Compute cost $c_{i,ins}$ for cluster $\{i\}$ using (\ref{eq:cost_ins}).
\If {$c_{i,ins} < c_{best}$}
        \State $c_{best} = c_{i,ins}$; $l_{best} = l$; $j_{best} = j$.
        \EndIf
\State Set cluster $C_j = \{k_1,\ldots,k_{M}\}$, where $M = |C_j|$.
\For {$1 \leq l \leq |C_j|+1$}
      \State Perform check $\textbf{C1}$.
      \State Compute pick-up and drop-off times $t_{p}^*$ and $t_{d}^*$ (see discussion).
      \If {$\mathrm{Pr}(\delta_{\max} \geq t_{p}^*, \gamma_{\max} \geq t_{d}^*) \geq 1 - \epsilon$}
        \State Compute cost $c_{i,ins}$ using (\ref{eq:cost_ins}).
        \If {$c_{i,ins} < c_{best}$}
        \State $c_{best} = c_{i,ins}$; $l_{best} = l$; $j_{best} = j$.
        \EndIf
       \EndIf
\EndFor
\State \Return $j_{best}$, $l_{best}$, $c_{best}$.
\EndProcedure
\end{algorithmic}  \caption{Insertion algorithm.}
\label{alg:insertion}
\end{algorithm}

With the insertion algorithm in hand, it is now possible to describe our cluster formation algorithm. The cluster formation algorithm searches through the current clusters for each passenger $i$ to determine whether the passenger can be inserted. After checking each potential cluster, passenger $i$ is inserted into the lowest cost and feasible cluster, and the next passenger is considered. Our cluster formation algorithm is detailed in Algorithm~\ref{alg:clusterformation}.

\begin{algorithm}[!t]
\begin{algorithmic}
\Procedure{ClusterFormation}{}
\State Randomly choose a unique index in $\mathcal{N} = \{1,2,\ldots,N\}$ for each passenger.
\State Initialize cluster $C_1 = \{1\}$;
\State Initialize allocation $\pi = \{C_1\}$ and partition index $\mathcal{I} = \{1\}$.
\State Initialize the set of unclustered passengers $\mathcal{U} = \{2,\ldots,N\}$.
\While{$\mathcal{U} \neq \emptyset$}
\State Set $\mathcal{J} \leftarrow \mathcal{I}$
\While {$\mathcal{J} \neq \emptyset$}
\State Randomly choose an element $j \in J$.
\State Update $j_{best}$, $l_{best}$, $c_{best}$ using Insertion($c_{best}$),$i$, $C_j$) (see Algorithm~\ref{alg:insertion}).
\State $\mathcal{J} \leftarrow \mathcal{J} \setminus j$
\EndWhile
\State Update $C_{j_{best}}$ by inserting $i$ before $l_{best}$; $\mathcal{U} \leftarrow \mathcal{U} \setminus i$; $\mathcal{J} \leftarrow \emptyset$.
\EndWhile
\State \Return Cluster allocation $\pi$
\EndProcedure
\end{algorithmic}  \caption{Cluster formation algorithm with probabilistic feasibility constraints.}
\label{alg:clusterformation}
\end{algorithm}

\textit{Passenger pricing:} With the cluster formation algorithm in hand, we now introduce the joint scheduling, routing, and pricing algorithm used by the service provider to generate expected profit maximizing offers. The algorithm approximately solves the following problem:
 \begin{equation}\label{eq:pricing_opt}
\begin{aligned}
& \underset{r,\epsilon}{\text{maximize}}
& & \mathbb{E}[P(r)] \\
& \text{subject to}
& & 0 \leq r \leq r_{u}\\
& & & 0 \leq \epsilon \leq 1\\,
\end{aligned}
\end{equation}
where $\mathbb{E}[P(r)]$ is the expected profit at price rate $r$, which is given by
\begin{align}
\mathbb{E}[P(r)] = \sum_{S \subset N} \left(\sum_{i \in S} p_i - c_i\right) \prod_{i \in S} \mathrm{Pr}(i~\mathrm{accept}) \prod_{j \in S^c} \left(1 - \mathrm{Pr}(j~\mathrm{accept})\right),
\end{align}
with $\mathrm{Pr}(i~\mathrm{accept})$ given by (\ref{eq:prob_accept}). Importantly, all price rates greater than $r_{u}$ are rejected by the passengers with probability one (this follows from (\ref{eq:beta_density})). The price offered to each passenger $i$ is then given by $p_i = rR_i$, where $R_i$ is the distance from passenger $i$'s pick-up to drop-off.

Our solution clusters passengers based on the probabilistic feasibility constraint with parameter $\epsilon_1 \approx 1$ and then optimizing the price using a standard scalar nonlinear optimization algorithm (i.e., a descent algorithm). This is repeated for parameter $\epsilon_{k+1} = \epsilon_{k} - \epsilon_{step}$ until $\epsilon_{k+1} < 0$, with $0 < \epsilon_{step} \leq 1$. The pricing and clustering solution that maximizes the expected profit over probabilistic feasibility parameters $\epsilon_1,\epsilon_2,\ldots$ is then chosen. As such, at the end of Stage 1, the price rate $r$ and passenger clusters $C_1,C_2,\ldots$ are obtained by the service provider. The procedure is summarized in Algorithm~\ref{alg:stage1}.

\begin{algorithm}[!t]
\begin{algorithmic}
\Procedure{ExpectedProfitMaximization}{}
\State Initialize $\epsilon$ and set $P_{opt} \leftarrow 0$.
\While {$\epsilon > 0$}
\State Set $\epsilon \leftarrow \epsilon - \epsilon_{step}$.
\State Solve (\ref{eq:pricing_opt}) to obtain the optimal expected profit $P^*$, with $\epsilon$ fixed.
\If {$P^* > P_{opt}$}
\State Set $P_{opt} \leftarrow P^*$; $r_{opt} \leftarrow r$.
\State Set $C_{opt} \leftarrow \{C_1,C_2,\ldots\}$ (corresponding to clusters from Algorithm~\ref{alg:clusterformation}).
\EndIf
\EndWhile
\State \Return $P_{opt},r_{opt},C_{opt}$.
\EndProcedure
\end{algorithmic}  \caption{Joint pricing and clustering algorithm to maximize the expected profit for Stage 1 of our market mechanism.}
\label{alg:stage1}
\end{algorithm}

The final output of the expected profit maximization algorithm is an offer consisting of a journey deviation and price for each passenger. This offer is then communicated to the passenger, and the service provider waits for the passengers response, given in the next stage.

\subsection{Stage 2: First Response}

In the second stage of the negotiation, each passenger makes a preliminary decision to accept or reject the \textit{conditional journey} offered by the service provider. If passenger $i$ accepts, it means that she has accepted the journey offer as long as the service provider does not change the offer. On the other hand, if passenger $i$ rejects the offer then she is no longer interested in a journey with the service provider.

We emphasize that if the passenger accepts, then she has agreed to a contract with the service provider; that is, she must pay for the service unless the service provider either raises the price or increases the journey deviation. We note that this type of contract is standard for other transportation services, such as pre-booked trains or buses.

\subsection{Stage 3: Re-Orientation}

In the third stage of the negotiation, the service provider has additional information. In particular, the service provider knows both what the users that conditionally accept are prepared to pay, and which users have rejected the offer.

As not all passengers will usually accept their offer, the service provider must update the passenger clusters. Although the cluster sizes will change if not all passengers accept, no passengers are allocated to different clusters. This ensures that the maximum journey deviation for each passenger does not change.

To obtain final prices for each passenger, the service provider solves the maximin profit problem over the passengers that have accepted in the previous stage. Let $S^*$ be the passengers that accepted their offers in the previous stage and $Q$ the final profit obtained from the market mechanism. The maximin profit problem is then
\begin{align}\label{eq:maximin}
\max_{\{r_i\}_{i \in S^*}} \min Q.
\end{align}
We note that in this stage of the mechanism, different passengers can be offered journeys at different price rates.

As the previous stages of the mechanism have revealed a lower bound on the maximum price each user is prepared to pay, the maximin profit problem in (\ref{eq:maximin}) is equivalent to finding the set of passengers
\begin{align}
S^*_c = \arg \max_{S_c \subset S^*} \sum_{k \in S_c} rR_k - c_{S_c},
\end{align}
and then pricing passengers such that the passengers in $S^*_c$ are charged $rR_k,~k \in S_c^*$, while the other passengers are charged at a higher price rate. This allows these passengers to find another service or for the provider to subcontract the journeys, which avoids losses.

The service provider uses the pricing strategy in Stage 3 to ensure that the passengers in desirable clusters accept, which in turn maximizes the service provider's profit. That is, the service provider will raise the price of passengers that have conditionally accepted so that they do not require the provider to be exposed to large losses. Although undesirable from the perspective of service provider reputation, we believe that this strategy is likely to be necessary in real-world practice. This is due to the fact that service providers have both physical (i.e, fleet size) and financial (i.e., initial capital) constraints. As such, it is not possible to service all passengers that might accept without either investing in a larger fleet size or hiring additional vehicles. To cope with these additional costs, it is necessary for passengers to be charged more when they are difficult to serve.

\subsection{Stage 4: Final Decision}

In the fourth stage of the negotiation, the remaining passengers make their final decision based on the latest offer from the service provider. Each passenger either \textit{unconditionally accepts} the final offer, or \textit{rejects} it; i.e, passengers that conditionally accepted in Stage 2 now either accept the offer or reject otherwise. We emphasize that a passenger cannot freely reject the offer unless the provider has increased the price; otherwise, the passenger must pay for the journey.

We point out that the service provider cannot always service all passengers that accept; either because there are not enough vehicles, or the passenger would cause a net loss for the provider. As such, it is highly desirable from the perspective of financial solvency of the provider to be able to raise the price and allow the passenger to find alternative transport. If the passenger still accepts the journey even after the price has been raised, then the service provider can use the additional revenue to hire an additional vehicle to service the passenger.

At the end of Stage 4 of the mechanism, all passengers to be serviced are known to the service provider, are priced, and have been allocated to vehicles. Moreover, each vehicle has a journey plan. The performance of our market mechanism is evaluated in Section~\ref{sec:simulation}. In the next section, we analyze the role of the time interval between mechanism runs, which determines the mechanism rate.

\section{Mechanism Parameter Design}

A key assumption in our market mechanism is that the passengers are known to the provider before the beginning of the mechanism. For dial-a-ride services this is known as the static scenario, and is often problematic as passengers can make requests after one run of the mechanism and before the next. This was solved by allowing dynamic arrivals, where passengers can be inserted while vehicles are on the road. However, the dynamic approach cannot be directly used with our mechanism without statistics for the locations and the prices dynamic passengers would be prepared to pay. While such an approach is in principle possible using historical passenger prices, and pick-up and drop-off location data, an unprecedented level of data refinement would be required.

To resolve this issue, we instead adapt the rate our market mechanism is run. Importantly, the mechanism rate is in fact a fundamental feature of any on-demand transport market mechanism. There are two key parameters that determine the mechanism rate: the probability that a passenger request is ignored; and the probability that a passenger cannot be serviced before the next mechanism run.

In this section, we derive simple analytical expressions for the probability a request is ignored, $P_{ignore}$, and the probability a passenger cannot be serviced in time, $P_{overtime}$. The key purpose of the analytical expressions are to guide design of the mechanism rate. In particular, we demonstrate the tradeoff between $P_{ignore}$ and $P_{overtime}$, as the time between mechanism runs is increased.

\subsection{Analysis}\label{sec:analysis}

Our analysis of $P_{ignore}$ and $P_{overtime}$ is based on a simplified probabilistic model of passengers and vehicles. Although the simplifications lead to a coarse approximation of real-world on-demand transport networks, conclusions from our analysis are supported by intuitive explanations. The key assumptions are as follows:
\begin{enumerate}
\item The time between mechanism runs is $T$ minutes and the corresponding rate is $R = 1/T$.
\item Each vehicle services only a single passenger per interval between mechanism runs. This is reasonable for short intervals $T$.
\item The time for a vehicle to travel a distance $z$ is given by $\nu z$, where $\nu$ is the average velocity.
\item The time each request arrives forms a homogeneous Poisson process with rate $\kappa$. Moreover, the time between request delivery and the desired pick-up time, $\Delta$, is exponentially distributed with mean $1/\lambda$.
\item Pick-up and drop-off locations are distributed according to a Poisson point process with intensity $\zeta$. Moreover, the drop-off location is the closest point to the pick-up locations, which means that the distribution of the distance is\footnote{This result follows immediately by considering the probability that there is no point within radius $R$ from the origin, which is given by $e^{-\pi R^2}$.}
    \begin{align}
      f_Z(z) = e^{-\zeta \pi z^2}2\pi \zeta z.
    \end{align}
\end{enumerate}

We now turn to analysis of $P_{ignore}$ and $P_{overtime}$. Our analysis is based on two new analytical expressions for the probabilities. First, the probability a request is ignored, $P_{ignore}$, is given in the following proposition.
\begin{proposition}\label{prop:prob_ignore}
The probability a request is ignored is given by
\begin{align}
  P_{ignore} &= 1 -\frac{1}{\lambda T}\left(1 - e^{-\lambda T}\right).
\end{align}
\end{proposition}
\begin{proof}
See Appendix~\ref{app:prob_ignore}.
\end{proof}

Observe that as $T \rightarrow \infty$, $P_{ignore} \rightarrow 1$, which means that as the interval between mechanism runs increases, the probability a request is ignored tends to one, which is consistent with intuition. On the other hand, as $T \rightarrow 0$, $P_{ignore} \rightarrow 0$.

The probability that a passenger cannot be serviced before the next mechanism run, $P_{overtime}$, is given in the following proposition.
\begin{proposition}\label{prop:long_service}
The probability a passenger cannot be serviced before the next mechanism run is given by
\begin{align}
  P_{overtime} &= \left(1 - \frac{2\pi\zeta T}{\nu}\right)e^{-\pi\zeta T^2/\nu^2} + \frac{\nu}{T\sqrt{\zeta}}\left(\frac{1}{2} - Q\left(\frac{T\sqrt{2\pi \zeta}}{\nu}\right)\right)
\end{align}
where $Q(\cdot)$ is the $Q$-function, a standard special function defined in (\ref{eq:Q_func}).
\end{proposition}
\begin{proof}
See Appendix~\ref{app:long_service}.
\end{proof}

Observe that as $T \rightarrow \infty$, $P_{overtime} \rightarrow 0$, and as $T \rightarrow 0$, $P_{overtime} \rightarrow 1$. This is an intuitive result as it simply states that when the interval is large the probability that a passenger cannot be serviced in time approaches zero. Importantly, this observation means that there is a tradeoff between $P_{ignore}$ and $P_{overtime}$. In other words, it is not possible to avoid passengers being ignored and not being serviced, with probability one using the same mechanism rate.

\subsection{Tradeoff}

\begin{figure}[!h]
\centerline{\includegraphics[height=4in]{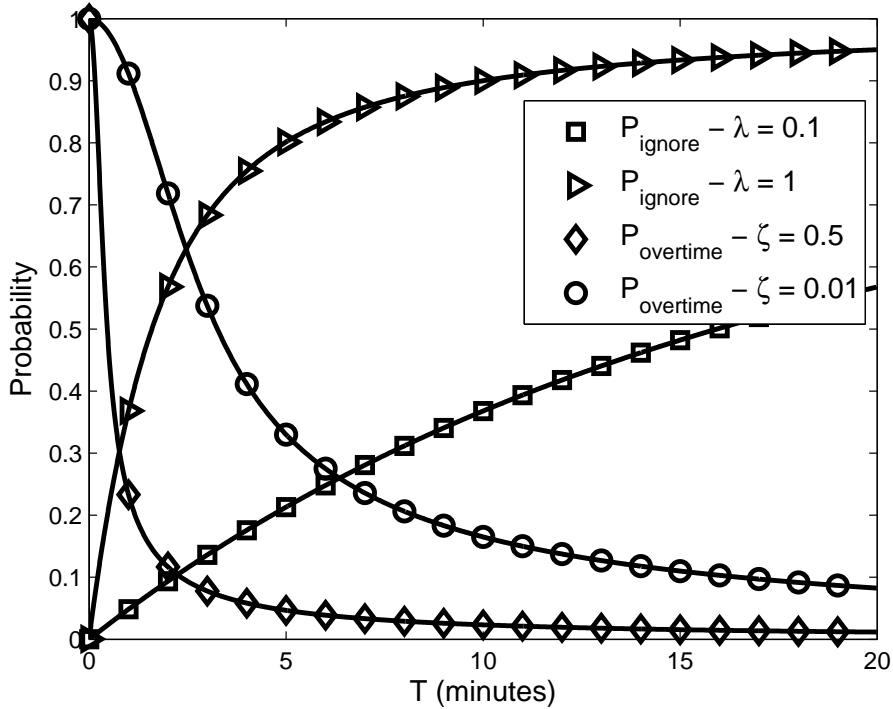}}
\caption{\textit{Plot of the tradeoff between $P_{ignore}$ (from Proposition~\ref{prop:prob_ignore}) and $P_{overtime}$ (from Proposition~\ref{prop:long_service}), for varying time between mechanism runs $T$. We assume vehicles travel at $20$ km/hr.}} \label{fig:tradeoff}
\end{figure}

It is not immediately obvious from Propositions~\ref{prop:prob_ignore} and \ref{prop:long_service}, which factors are key in determining the tradeoff between $P_{ignore}$ and $P_{overtime}$. As such, we now examine the behavior of $P_{ignore}$ and $P_{overtime}$ numerically.

Fig.~\ref{fig:tradeoff} plots the tradeoff between $P_{ignore}$ and $P_{overtime}$ for varying $\lambda$ and $\zeta$, based on our analysis in Section~\ref{sec:analysis}. Observe that increasing $\lambda$, also leads to an increase in $P_{ignore}$. On the other hand, an increase in $\zeta$ leads to increase in $P_{overtime}$. As such, the intersection between $P_{ignore}$ and $P_{overtime}$ (the crossover rate) reduces when $\lambda$ or $\zeta$ are increased.

\section{Simulation Results}\label{sec:simulation}

With our mechanism for the new thread of on-demand transport services in hand, we are ready to present an initial business case for our mechanism. The benchmark we use for comparison is the fixed price-rate approach, where passengers are charged the same price-rate, irrespective of the number of potential passengers or the requested journeys. In particular, under the fixed price-rate policy, each passenger is charged at a rate given by the expected maximum price-rate passengers are prepared to pay, given by $\mathbb{E}[r_{i,\max}]$, which can be easily obtained from (\ref{eq:beta_density}).

The remainder of this section consists of two parts. First, we compare our mechanism as detailed in Section~\ref{sec:proposed_approach} with the standard fixed price-rate policy. The performance of our mechanism is evaluated in terms of both the expected profit and the expected efficiency (based on Definition~\ref{def:efficiency}). We demonstrate key trends in these performance metrics with variations in the number of potential passengers, passenger demand, and the clustering parameter $\epsilon$ (see Definition~\ref{def:eps_feasible}). Second, we evaluate the performance of a modification of our mechanism in a scaled-up system with up to $100$ passengers. We describe this modification and demonstrate the performance as the number of potential passengers are varied, in comparison with the fixed price-rate approach.

\subsection{Key Trends}\label{sec:key_trends}

We now illustrate key trends in the expected profit and expected efficiency, based on a network setup with $K=5$ drivers and up to $N = 13$ potential passengers. Importantly, this setting is practical---despite the small scale---when the serviced region is partitioned and separate negotiations are performed in each partition. This is one approach that reduces the need for suboptimal heuristics for pricing, which can ensure practical negotiation run-times. We evaluate an alternative approach based on a heuristic approach on a
scaled-up network in Section~\ref{sec:large_scale}.

The pick-up and drop-off locations of passengers in the network are drawn from real locations in Prague, Czech Republic: $K = 5$ drivers; average vehicle velocity $\nu = 30$ km/hr; cost/km of $0.4$ euros/km; and a maximum price-rate for each passenger of $3$ euros/km. We further assume that there is an hour between each mechanism run, which means that the beginning of a passengers pick-up interval is uniformly distributed over the $60$ minute interval. The maximum length of each passenger's pick-up interval is $10$ minutes, with the actual interval length uniformly distributed.

We first demonstrate the performance mechanism with clustering using hard constraints (i.e., we set $\epsilon = 0$ and $\epsilon_{step} > 1$). We note that the mechanism with a fixed price-rate given by $\mathbb{E}[r_{i,\max}]$ and clustering with hard constraints is used as a benchmark.

Fig.~\ref{fig:profitfixed} plots the expected profit per mechanism run for a varying number of potential passengers, $N$. The key observation is that our mechanism with optimized price-rate always improves the profit over the fixed price-rate approach, although the improvement depends on the passenger demand; i.e., how much passengers are willing to pay. Observe that there is a significant improvement in expected profit for both low and high demand, corresponding to $\alpha_r = 1,\beta_r = 3$ and $\alpha_r = 1, \beta_r = 1$, respectively. We also point out that the rate of increase in the expected profit with the number of passengers, $N$, is reduced as $N$ increases. This is due to a saturation effect, where increasing the number of potential passengers in the network does not significantly improve the profit.

\begin{figure}[!h]
\centerline{\includegraphics[height=4in]{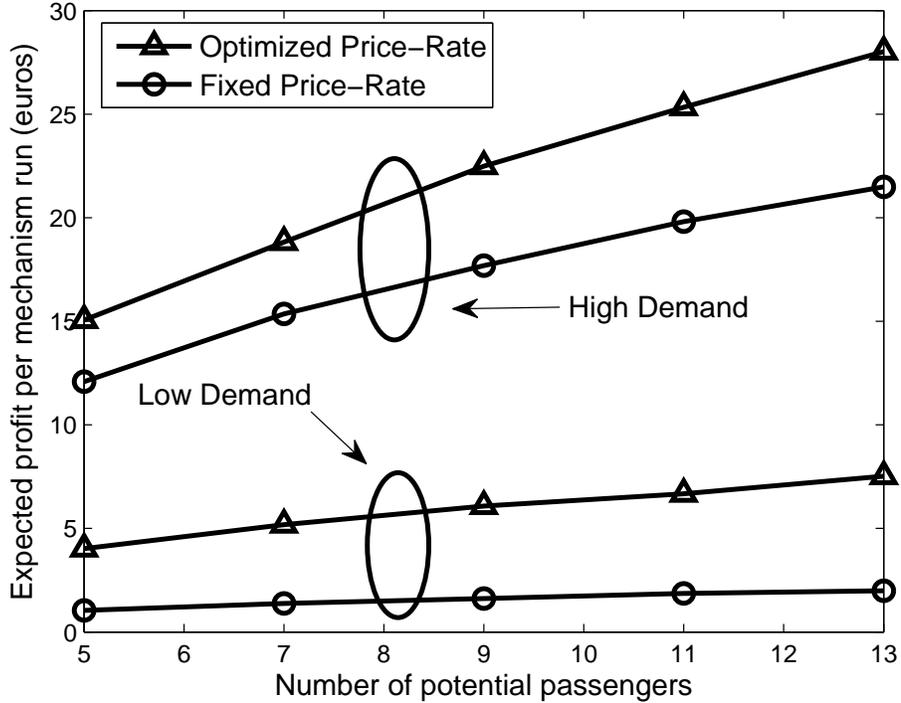}}
\caption{\textit{Plot of the expected profit for our mechanism with optimized price-rate and the standard fixed price-rate approach, for varying number of potential passengers, $N$. High demand corresponds to $\alpha_r = 3$, $\beta_r = 1$, medium demand corresponds to $\alpha_r = \beta_r = 1$, and low demand corresponds to $\alpha_r = 1$, $\beta_r = 3$. Simulation parameters: $K = 5$; average velocity of $30$ km/hr; cost/km of $0.4$; maximum price-rate of $3$ euros/km; and hard feasibility constraints are enforced.}} \label{fig:profitfixed}
\end{figure}

Next, we turn to the expected efficiency (see (\ref{def:efficiency})) of our mechanism demonstrated in Fig.~\ref{fig:efficiencyfixed}. In this case, observe that our mechanism improves the expected profit over the fixed price-rate approach for low demand with $N \geq 7$. This means that our mechanism outperforms the fixed price-rate approach in both expected profit and efficiency. On the other hand, the fixed price-rate approach has a higher efficiency for high demand.

Next, we demonstrate the performance improvements that can be obtained by using our optimized clustering algorithm, where $\epsilon_{step} \leq 1$. In particular, this improves the efficiency in the case of high demand.

\begin{figure}[!h]
\centerline{\includegraphics[height=4in]{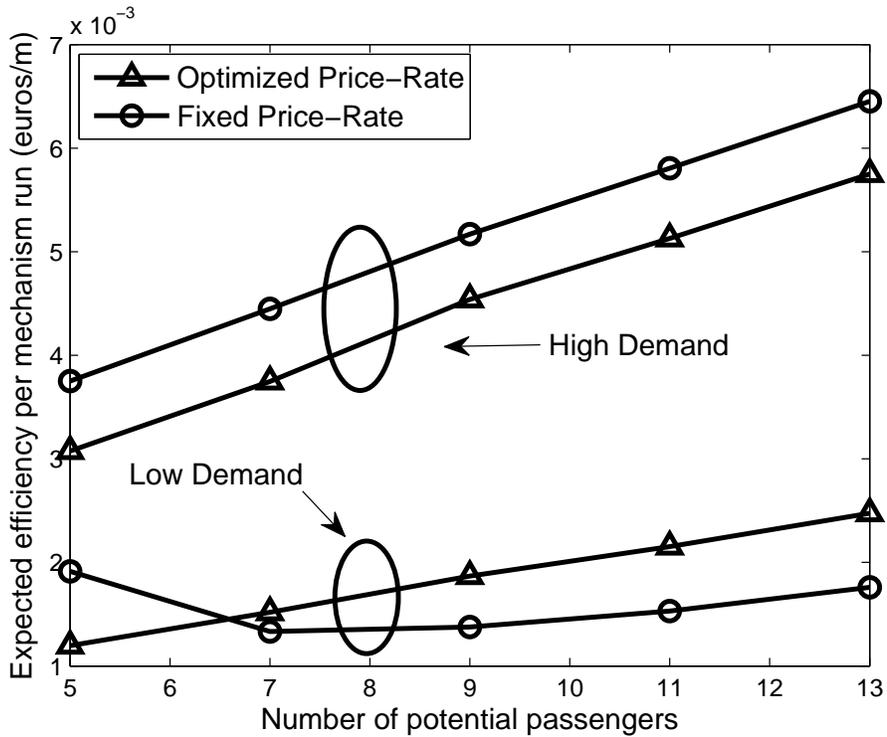}}
\caption{\textit{Plot of the expected efficiency for our mechanism with optimized price-rate and the standard fixed price-rate approach, for varying number of potential passengers, $N$. High demand corresponds to $\alpha_r = 3$, $\beta_r = 1$, medium demand corresponds to $\alpha_r = \beta_r = 1$, and low demand corresponds to $\alpha_r = 1$, $\beta_r = 3$. Simulation parameters: $K = 5$; average velocity of $30$ km/hr; cost/km of $0.4$ euros/km; maximum price-rate of $3$ euros/km; and hard feasibility constraints are enforced.}} \label{fig:efficiencyfixed}
\end{figure}

In Table~\ref{table:optclustering_profit}, the expected profit is compared with the number of potential passengers, with and without optimized clustering. The optimized clustering solves (\ref{eq:pricing_opt}) with $\epsilon_{step} = 0.2$, while the algorithm without optimized clustering enforces the hard constraints (i.e., $\epsilon = 0$). In both cases, the price-rate is optimized. Observe that our optimized clustering algorithm improves the expected profit, even with optimal pricing. In particular, gains of up to $4$ euros per negotiation can be achieved. This means that the expected profit for the fixed price-rate approach can be improved significantly (up to $10$ euros per negotiation with $\alpha_r = \beta_r = 1$) by using both optimized pricing and clustering.

\begin{table}[!h]
\caption{Expected profit per mechanism run (in euros) with and without optimized clustering. Parameters: $\alpha_r = \beta_r = 1$, $\alpha_{\delta} = 3$, $\beta_{\delta} = 1$.}\label{table:optclustering_profit}
\begin{center}
 \begin{tabular}{|c|c|c|c|c|c|}
 \hline
 Potential Passengers $N$ & $5$ & $7$ &  $9$ & $11$ & $13$\\ \hline
 With Optimized Clustering (euros) & $17.5$ & $23.2$ & $26.3$ & $28.1$ & $30.0$ \\ \hline

 Without Optimized Clustering (euros) & $14.5$ & $18.8$ & $22.0$ & $25.5$ & $27.8$ \\ \hline
 \end{tabular}
\end{center}

\end{table}

In Table~\ref{table:optclustering_eff}, the expected efficiency (in euros/m) is compared with the number of potential passengers, with and without optimized clustering. Observe that the expected efficiency of optimized clustering outperforms the expected efficiency without optimized clustering. Importantly, optimizing the clustering ensures that the expected efficiency of our mechanism is comparable with the efficiency using the fixed price-rate approach. As such, our mechanism can outperform the fixed price-rate approach in both expected profit and expected efficiency.

\begin{table}[!h]
\caption{Expected efficiency per mechanism run (in euros) with and without optimized clustering. Parameters: $\alpha_r = \beta_r = 1$, $\alpha_{\delta} = 3$, $\beta_{\delta} = 1$.}\label{table:optclustering_eff}
\begin{center}
 \begin{tabular}{|c|c|c|c|c|c|}
 \hline
 Potential Passengers $N$ & $5$ & $7$ &  $9$ & $11$ & $13$\\ \hline

 Optimized Clustering (euros/m) & $0.0038$ & $0.0051$ & $0.0063$ & $0.0074$ & $0.0085$ \\ \hline
 Without Optimized Clustering (euros/m) & $0.0030$ & $0.0038$ & $0.0044$ & $0.0052$ & $0.0057$ \\ \hline

 \end{tabular}
\end{center}

\end{table}

\subsection{The Effect of Network Scaling}\label{sec:large_scale}

We now turn to simulation results for a scaled-up network with $30$ vehicles and up to $100$ passengers. The average velocity, cost/km, maximum price-rate, and pick-up interval distributions are the same as in Section~\ref{sec:key_trends}.

In order to apply our mechanism to a network of this scale, it is necessary to modify our mechanism by approximating the objective of the optimization problem in (\ref{eq:pricing_opt}) by truncating the sum involved in the expectation in (\ref{eq:pricing_opt}). Moreover, clustering is limited to the finding the first feasible cluster instead of the full search detailed in Algorithm~\ref{alg:insertion}. Using these heuristics, next we compare our mechanism with the fixed price-rate approach without clustering; i.e., each vehicle serves up to one passenger per negotiation run.

Table~\ref{table:large_scale} demonstrates the effect of the heuristic on the performance of the network. Observe that in all cases the optimized pricing in our mechanism outperformed the fixed price-rate approach, in terms of expected profit. Note that the expected profit for the fixed price-rate approach is approximately constant, irrespective of the number of passengers. This is due to the saturation effect also observed in Fig.~\ref{fig:profitfixed}. On the other hand, the expected profit for the optimized pricing approach reduces as the number of potential passengers increases, in contrast with the results in Section~\ref{sec:key_trends}. This is due to the approximation of the objective, which is less accurate as the number of potential passengers increases.

\begin{table}[!h]
\caption{Expected profit per mechanism run (in euros). Parameters: $\alpha_r = \beta_r = 1$, $\alpha_{\delta} = 3$, $\beta_{\delta} = 1$.}\label{table:large_scale}
\begin{center}
 \begin{tabular}{|c|c|c|c|}
 \hline
 Potential Passengers $N$ & $30$ & $60$ & $90$\\ \hline
 Optimized Pricing (euros) & $83.7$ & $66.7$ & $52.9$\\ \hline
 Fixed Price-Rate (euros) & $52.6$ & $51.7$ & $52.4$\\ \hline
 \end{tabular}
\end{center}

\end{table}

\subsection{Summary Of Simulation Results}

We summarize the key results from our simulation study as follows.
\begin{enumerate}
\item Our mechanism outperforms the fixed price-rate approach in terms of expected profit, even without optimized clustering (i.e., $\epsilon = 0$ and $\epsilon_{step} > 1$) or using heuristic approximations.
\item Using optimized clustering with $\epsilon_{step} \leq 1$ both increases the expected profit and the expected efficiency of our mechanism. Moreover, the expected profit is significantly higher than the fixed price-rate approach, irrespective of the passenger demand, and also has comparable efficiency.
\item Increasing the passenger demand (i.e., varying the parameters $\alpha,\beta$ in the distribution) leads to an increase in the expected profit.
\item Increasing the number of potential passengers using our optimized clustering algorithm increases the expected profit; however, the expected profit saturates at a large number of potential passengers.
\item Using the heuristic modification of our mechanism for scaled-up networks (described in Section~\ref{sec:large_scale}) outperforms the fixed price-rate approach in terms of expected profit. However, the expected profit reduces as the number of potential passengers is increased. This suggests that an appropriate method to deal with large-scale networks is to partition the network by simultaneously running several small-scale negotiations. In this case, approximations of the objective function are not required, which in turn ensures that the expected profit can increase as the number of potential passengers increases.
\end{enumerate}

\section{Conclusions}

We have proposed a new market mechanism for a new thread of on-demand transport services, which enables negotiations with passengers to both increase provider profits and select passengers that value the service the most. A key feature of our mechanism is that it jointly optimizes scheduling, routing and passenger pricing; in sharp contrast with standard approaches for services targeted at the elderly and disabled, and taxis.

Our mechanism is based on a new agent-based model, which emphasizes the role of provider profit in allocating resources. In particular, we have developed new models that incorporate price-based preferences for both the passengers and the provider. We also consider the effect of deviations from passenger requests in our resource allocation.

Our simulations results for new thread of on-demand transport services demonstrated that our market mechanism improves the profitability of the service provider, compared with standard fixed price-rate approaches, while maintaining comparable efficiency.

\section*{Acknowledgments}

Access to computing and storage facilities owned by parties and projects contributing to the National Grid Infrastructure MetaCentrum, provided under the programme "Projects of Large Infrastructure for Research, Development, and Innovations" (LM2010005), is greatly appreciated.

\appendices

\section{Proof of Proposition}\label{app:prob_ignore}

Consider a time interval $[nT,(n+1)T]$ ($n \in \mathbb{N}$) between two consecutive mechanism runs, with $a \in [nT,(n+1)T)$. Then, the probability that a passenger's request is ignored is given by
\begin{align}
P_{ignore} &= 1 - \mathrm{Pr}(a + \Delta \geq (n+1)T).
\end{align}

Since request arrive according to a homogeneous Poisson process, the following lemma gives the distribution of the earliest pick-up time, $a$.
\begin{lemma}\label{lem:a_dist}
The earliest possible pick-up time $a$ is uniformly distributed in $[nT,(n+1)T)$, conditioned on the number of requests in the interval.
\end{lemma}
This means that for a random request in the interval $[nT,(n+1)T)$, $a$ is uniformly distributed.

Using Lemma~\ref{lem:a_dist} and conditioning on the time between the request arrival and desired pick-up time, $\Delta$, we have
\begin{align}
P_{ignore} &= \int_0^{\infty} \mathrm{Pr}(a \geq (n+1)T - \delta|\Delta = \delta)\lambda e^{-\lambda \delta}d\delta.
\end{align}
Integrating by parts, we obtain the required result.

\section{Proof of Proposition}\label{app:long_service}

As for Proposition~\ref{app:prob_ignore}, consider a time interval $[nT,(n+1)T)$ ($n \in \mathbb{N}$) between two consecutive mechanism runs, with $a \in [nT,(n+1)T)$. Then, the probability that a passenger cannot be serviced before the next mechanism run is given by
\begin{align}
  P_{overtime} &= \mathrm{Pr}(a + \nu z \geq (n+1)T)\notag\\
  &= \int_0^{\infty} \mathrm{Pr}(a \geq (n+1)T - \nu z)f_Z(z)dz,
\end{align}
which follows by conditioning on the distance, $Z$. We note there is no loss in generality by assuming that the pick-up location is at the origin by Slivnyak's theorem \cite{Baccelli2009}.

Using Lemma~\ref{lem:a_dist} and the fact that $a \in [nT,(n+1)T)$, it follows that
\begin{align}
  P_{overtime} &= \frac{\nu}{T}\int_0^{T/\nu} zf_Z(z)dz + \int_{T/\nu}^\infty f_Z(z)dz\notag\\
  &= \frac{\nu}{T}\int_0^{T/\nu} 2\pi \zeta z^2 e^{-\pi \zeta z^2} dz + e^{-\pi\zeta T^2/\nu^2}.
\end{align}
Next we integrate by parts\footnote{We use the fact that $\int 2\pi \zeta ze^{-\pi \zeta z^2}dz = -e^{-\pi \zeta z^2}$.} to obtain
\begin{align}
  P_{overtime} &= \frac{\nu}{T}\left(\frac{T}{\nu}e^{-\pi\zeta T^2/\nu^2} + \int_0^{T/\nu} e^{-\pi \zeta z^2}dz\right) + e^{-\pi\zeta T^2/\nu^2}.
\end{align}
We then apply the change of variables $u = \sqrt{2\pi\zeta}z$ and the definition of the $Q$-function, which is given by
\begin{align}\label{eq:Q_func}
  Q(x) = \frac{1}{\sqrt{2\pi}}\int_x^\infty e^{-u^2/2}du.
\end{align}
Using the fact that $Q(0) = \frac{1}{2}$, we then have
\begin{align}
  P_{overtime} &= \left(1 - \frac{2\pi\zeta T}{\nu}\right)e^{-\pi\zeta T^2/\nu^2} + \frac{\nu}{T\sqrt{\zeta}}\left(\frac{1}{2} - Q\left(\frac{T\sqrt{2\pi \zeta}}{\nu}\right)\right),
\end{align}
as required.

\bibliographystyle{ieeetr}
\bibliography{ondemand}

\end{document}